%% file: socialcapital-arXiv.tex
\RequirePackage{fix-cm}

\documentclass[natbib,smallcondensed]{svjour3}     % onecolumn (ditto)

\smartqed  % flush right qed marks, e.g. at end of proof

\usepackage{float}
\usepackage[caption = false]{subfig}
\usepackage{algorithm,algpseudocode}
\usepackage{booktabs}
\usepackage[]{graphicx}
\usepackage{rotating}
\usepackage{multirow}
\usepackage{amsmath,amssymb}
\usepackage{enumerate}
\usepackage{ulem}
\usepackage{color}
\usepackage{capt-of}

\newtheorem{Proposition}{Proposition}
\begin{document}
\title{Social Centrality using Network Hierarchy and Community Structure}
%\titlerunning{Short form of title}        % if too long for running head

\author{Rakhi Saxena  
       \and  Sharanjit Kaur 
       \and  Vasudha Bhatnagar
}

%\authorrunning{Short form of author list} % if too long for running head

\institute{Rakhi Saxena \at
              Deshbandhu College, 
              University of Delhi\\
               \email{rsaxena@db.du.ac.in}           %  \\
%             \emph{Present address:} of F. Author  %  if needed
           \and
           Sharanjit Kaur \at
             Acharya Narendra Dev College, 
              University of Delhi\\
               \email{sharanjitkaur@andc.du.ac.in}
           \and
           Vasudha Bhatnagar \at
             Department of Computer Science, 
              University of Delhi\\
               \email{vbhatnagar@cs.du.ac.in}  
}

\date{Received: date / Accepted: date}

% The correct dates will be entered by the editor
\maketitle
\begin{abstract}
Several centrality measures have been formulated to quantify the notion of `importance' of actors in social networks. Current measures scrutinize either local or global connectivity of the nodes and have been found to be inadequate for social networks.  Ignoring hierarchy and community structure, which are inherent in all human social networks, is the primary cause of this inadequacy. Positional hierarchy and embeddedness of an actor in the community are intuitively crucial determinants of his importance. 

The theory of social capital asserts that an actor's importance is derived from his position in network hierarchy as well as from the potential to mobilize resources through intra-community (bonding) and inter-community (bridging) ties. Inspired by this idea, we propose a novel centrality measure SC (Social Centrality) for actors in social networks.  Our measure accounts for - i) an individual's propensity to socialize, and ii) his connections within and outside the community. These two factors are suitably aggregated to produce social centrality score.  

Comparative analysis of SC measure with classical and recent centrality measures using large public networks shows that it consistently produces more realistic ranking of nodes. The inference is based on the available ground truth for each tested networks. Extensive analysis of rankings delivered by SC measure and mapping with known facts in well-studied networks justifies its effectiveness in diverse social networks. Scalability evaluation of SC measure justifies its efficacy for real-world large networks.
\end{abstract}
\section{Introduction}
Centrality is widely-used for identifying important/powerful nodes in a network \citep{Bloch16, Landherr2010}.  Node centrality is based on the perception of importance and hence is subjective. Different centrality measures have been devised and analysed by social scientists, computer scientists, mathematicians, physicists etc. \citep{Bloch16, survey-BoldiV13, BorgattiE06}. Theoretical and empirical investigations of these measures include surveys \citep{survey-BoldiV13, Landherr2010} and studies on their correlation \citep{li2015,correlated-centrality}. Since each centrality measure is effective in limited contexts \citep{Bloch16, KangPST11}, selecting a suitable centrality measure is a challenging task in social network analysis. In this setting, we propose a novel centrality measure for social networks, which is inspired by the well-studied {\it Sociology}  theory of {\it Social Capital} \citep{Lin2008, Putnam2002}. 
\subsection{The Problem and Motivation}
Despite the existence of several exact and approximate node centrality measures,  \textit{social importance} of actors in social networks is inadequately captured. When applied to social networks, prevalent centrality measures fall short of delivering ranks consonant with \textit{social importance} of actors enumerated by ground truth. Existing centrality measures consider direct benefits from \textit{actual resources} owned by an actor (number of connections, number of shortest paths on which an actor lies, number of connections that are reachable in $k$--hops, etc.) to gauge his importance. According to the theory of Social Capital,  an actor's social value is derived not only from the resources he owns but also from his \textit{potential} to acquire resources through social interactions. Researchers have asserted that availability of potential resources is dependent on an actor's position in network hierarchical and  community structure \citep{Lin2008,Putnam2002}. 

Though existing centrality measures are effective in diverse networks, they are inadequate to capture centrality in social networks. Taking into consideration the social capital theory, we ascribe following two reasons for this shortcoming:
\begin{enumerate}[(i)]
\item
These methods are oblivious to actor's position in network \textit{hierarchy} that is instrumental in determining the \textit{ease} of an actor to mobilize resources. 
\item
These methods ignore actor's embeddedness in network \textit{community structure} that determines the \textit{amount and variety} of resources available to the actor. 
\end{enumerate}
Our study on impact of network hierarchy and community structure  on social importance results in a novel measure called \textit{Social Centrality} (SC). The idea is inspired by the established {\it Theory of Social Capital} in Sociology \citep{Lin2008,Putnam2002}.  We show that SC is effective in diverse human social networks such as collaboration, terrorist and email networks.  Simple computation of SC makes it easily amenable to parallelization, and hence  scalable.
\subsection{Theory of Social Capital}
\label{social-cap-th-sec}
The potential of an actor to generate resources is termed as \textit{social capital} and has been a topic of substantial interest among Sociology community. Nahapiet and Ghoshal define social capital as `sum of actual and potential resources embedded within, available through, and derived from the network possessed by an individual' \citep{ghoshalsc-98}. 

Lin's structural postulate for social capital states that actors in a network form a pyramidal hierarchy in terms of distribution of resources and an actor's position in the hierarchical structure determines the availability of those resources \citep{Lin2008}. Since actors access resources via ties with neighbours, quality of ties influences ease of access of resources. The notion of strength of ties was introduced in the seminal article by Granovetter \citep{granovetter73}.  He argued that strong ties tend to \textit{bond} similar people and promote mutual connections leading to a \textit{tightly-knit groups}. Weak ties, on the other hand, connect heterogeneous groups forming bridges in social networks. Putnam premised that \textit{bonding} and \textit{bridging} ties  are primary determinants of social capital of an actor \citep{Putnam2002}.Weak or bridging ties are crucial for social capital because they bring together dissimilar individuals owning diverse resources. 

The notion of \textit{tightly-knit group} of \textit{homogeneous} actors in Sociology translates into  \textit{community} in Social Network Analysis. Communities in a network are subsets of nodes 
that are relatively densely connected to each other but sparsely connected to other dense subsets in the network.
It is well documented that community structures are inherent in social networks \citep{Wang15, Girvan02}. Since edges appear with a high concentration within communities and sparsely between communities, differently embedded actors have access to a variety of resources.
\subsection{Our Contributions}
In this paper, we introduce \textit{Social Centrality} (SC) as novel centrality measure that exploits network hierarchy and community structure in social networks. The summary of contributions follows.
\begin{enumerate}[(i)]
\item We use the theory of social capital to compute  a numeric measure denoting the importance of an individual in a social network (Sec. \ref{sc-sec}). 
\item We use k-truss decomposition method in an innovative manner to approximate community structure and determine nature of ties of actors. We also use trussness property to elicit hierarchy of nodes in the network (Sec. \ref{truss-sec}, \ref{sc-sec}).
\item We perform extensive evaluation of SC on real-world networks that sport ground truth and compare its effectiveness with classical and recent centrality measures using real-world networks  (Sec. \ref{sec-validation} to \ref{sec-quality}).
\item We demonstrate empirically the marginal loss in performance and significant gain in scalability by using trussness based approximation of communities instead of full-blown community detection to reveal community structure (Sec. \ref{sec-consonance}, \ref{sec-scalability}).
\item We evaluate the scalability of SC w.r.t. large synthetic and real-world networks (Sec. \ref{sec-scalability}).
\end{enumerate}
\section{Related Work}
\label{sec-related}
Centrality measures are broadly classified into two groups namely, degree-based and flow-based measures \citep{KangPST11, BorgattiE06}. 

Degree-based measures compute node importance on the basis of immediate neighbours. Degree Centrality (DC), counts all one-step neighbours completely ignoring global network topology. Eigenvector Centrality (EC), Principal Component Centrality (PCC) are extensions of DC. EC captures the importance of neighbours recursively and is computed as the principal eigenvector of the adjacency matrix defining the network \citep{bonacich1987power}. PCC  identifies multiple sets of influential neighbourhoods compared to single influential neighbourhood identified by EC \citep{IlyasR10}. Laplacian Centrality (LC) uses structural information about connectivity and density around the node to quantify its importance \citep{Laplacian13}. 

Flow-based measures utilize the notion of walk lengths and walk counts. Closeness Centrality (CC) falls under this category and measures node centrality as the inverse of the sum of geodesic distance from all other nodes \citep{Freeman:SN79}. Betweenness Centrality (BC) computes importance of node on the basis of the number of shortest paths passing through it \citep{Freeman:SN79}. Recently proposed  Spanning Tree Centrality (STC) measures the vulnerability of nodes in keeping the network connected \citep{Qi2015}. 

Researchers have also attempted to capture node importance by measuring social capital. Network Constraint (NC) quantifies social capital by measuring the concentration of an actor's connections in a single group of interconnected neighbours \citep{burt2001structural}. Recently, Subbian et al. proposed {\it SoCap} measure that computes social capital of the network and distributes it among nodes using a value-allocation function \citep{SoCap2014}. Based on the enumeration of all-pair-shortest-path, SoCap is computationally expensive.  Both of these methods exclude benefits arising from the hierarchical position and community membership of the individual.

Recently Gupta et. al. proposed a centrality measure that uses community detection to detect community structures \citep{Gupta2016}. The proposed method Comm Centrality (CoC) is a weighted combination of a node's inter- and intra-community links.  Intra-community links of the nodes attract higher weights than inter-community links while calculating their centrality value. However, CoC is oblivious to the hierarchical structure of real-world networks. 

\textit {Existing methods for discovering important nodes do not take cognizance of both hierarchy and community structure in social networks for determining centrality of actors. Since humans derive benefits concomitant with their position in the network hierarchy, and with the strength of their intra-- and inter--community connections, we posit that a centrality measure that takes these aspects into account gauges the importance of individuals more realistically in social networks. }
\section{Preliminaries}
\label{prelims-sec}
In this section we present the formal notation used in the paper, and sociology concepts that form the basis of the proposed centrality measure.  

We represent a social network as simple, undirected, unsigned,  edge-weighted graph $G$=$(V,E,W)$, a triplet formed by (i)  finite set of nodes  $V$, (ii) set of edges  $E \in V \times V$ , and (iii) $W:V \times V \rightarrow \mathbb{R}_{\ge 0}$, an edge-weight matrix. Let $|V| \,(= n)$ be the number of nodes  and $|E|\, (=m)$  be the number of edges. $V$ in $G$ models social actors, edge  $e_{ij} \in E$ models relation or link between  actors $v_i$ and $v_j$ $(v_i, v_j \in V)$. Weight $w_{ij} \in W$, of edge $e_{ij}$ quantifies the extent of relationship between $v_i$ and $v_j$. If $e_{ij} \notin E$, then $w_{ij}=0$. Let $N_i$ denote the set of neighbours of vertex $v_i$. $N_i \cap N_j$ is the set of common neighbours of vertices $v_i$ and $v_j$.
\subsection{Hierarchy in Networks}
Hierarchical decomposition is a pragmatic approach to analyse, visualize and understand massive networks  \citep{ Alvarez08, kcoretool06}. A hierarchy in a network $G$ is a partition of vertex set $V$  based on a structural node property function. Formally,
\begin{definition}
\uline{Hierarchy:} Let $G$=$(V,E,W)$ be a graph and $P$ be a function defined on a structural node property such that ($P:V\rightarrow\mathbb{N}^+$).  A hierarchy $H_P$ on $G$ is a partition of $V$ into $k$ subsets $\{V_1, \dots V_k\}$, such that 
$\forall v \in V_l, P(v) = \xi_l, l = 1, \dots, k $.
\end{definition}
Clearly we may assume that $\xi_l < \xi_{l+1}$. All nodes in $V_l$, share the same property value $\xi_l$. This ordering is mapped to the hierarchy levels in the network $G$. Nodes in $V_1$ are considered to be at level $l_1$, which is lower than that of the nodes in $V_2$ (i.e. $l_2$), and so on.  Example of structural node properties that can be used to derive hierarchy include node degree, number of nodes reachable in k-hops, number of shortest paths on which the node lies, number of closed triads the node is part of, etc. It is straightforward to extend the idea to property values that evaluate to $\mathbb{R}^+$.

\subsection{Social Importance, Hierarchy and Community Structure}
Sociologists assert that the importance of a social actor is a consequence of his direct ties as well as embeddedness in network structure \citep{Lin2008, Adler2002Social, ghoshalsc-98}. An actor's immediate personal ties provide the mechanism to acquire resources quickly \citep{Leonid14, Lin2008, ghoshalsc-98}, while embeddedness in network structure enhances the actor's ability to acquire additional resources \citep{Adler2002Social, ghoshalsc-98}. Hierarchy is an important dimension of a network that shapes the structure of social relations and determines the ease of access to resources via inter-personal linkages \citep{Lin2008}. Structural features also give rise to communities, which manifest as the subset of actors sharing semantically close relationship and influences the actor's ability to access resources within and outside the community to which actor belongs \citep{Adler2002Social}. 

In summary, embeddedness in hierarchy and community structure shapes the ease, the amount and the range of resources that an actor can potentially acquire, while direct ties determine how much of this potential will be actualized \citep{Leonid14}.  
\subsection{Bonding and Bridging Potential}
Every individual in a social network has potential to \textit{bridge} and \textit{bond} with other members of the network. Structural embeddedness of an actor determines his potential for bonding with members of his community and bridging with outsiders \citep{Leonid14, Adler2002Social}. Intra-community ties of an individual have a higher proportion of shared links (neighbours) within community \citep{David2010}. Such ties naturally give rise to cohesive regions in the network and result in bonding potential of the individual. Inter--community ties, on the other hand, have fewer common neighbours. Interestingly, such ties are important because they bridge communities and allow fresh ideas to percolate among community members \citep{granovetter73}. Thus inter--community ties determine the bridging potential of an individual. 

Quantification of bridging and bonding potential of an actor in a network requires knowledge of intra- and inter-community ties. In absence of a-priori information, the only recourse is to detect community membership and determine nature of ties. Except for label propagation algorithm with its known weakness of instability \citep{LabelRank13}, most community detection algorithms are computationally prohibitive for large networks \citep{Wang15}. The challenge in computing social centrality efficiently for large  networks is to determine the nature of actors' ties without running a community detection algorithm.

\subsection{Discriminating between Nature of Ties}
To circumvent the need of applying community detection algorithm, we adopt the observation made by Shi et. al. \citep{Adamic07} that close contacts of an individual themselves tend to know each other. This is natural because of relatively frequent opportunities for interactions among individuals. This situation is modelled topologically by three nodes linked to each other, forming a closed triad. Typically, ties within communities have more common neighbours and hence are part of a larger number of closed triads \citep{David2010, Adamic07}. In contrast, ties between communities have fewer common neighbours and are part of a 
lesser number of closed triads \citep{David2010, Adamic07}. 

Bolstered by these observations, we ascertain the nature of ties based on their participation in closed triads. Previous works \citep{Nick2013, Radicchi2004} distinguish between intra- and inter-community ties based on the count of closed triads tie is part of. Since we need to ascertain both node position in network hierarchy as well as nature of ties we use the k-truss hierarchical decomposition method \citep{Cohen08}. This method peels a graph into nested, dense subgraphs composed of closed triads. The trussness of an edge is a stronger indicator of belongingness to the same community because k-truss decomposition promotes local feature of an edge (number of closed triads it participates in) into a global feature of the edge (trussness). {\it Trussness} of edges is used to determine node trussness, which is employed to ascertain network hierarchy. The nature of tie is determined by the trussness of the node pair and their connecting edge. The details follow in the next section. 

%\rakhi{Nick et al. \citep{Nick2013} use the concept of Simmelian tie strength (count of triads tie is part of) to infer intra-community edges. Radicchi et al. \citep{Radicchi2004} introduce the edge clustering coefficient defined as the number of triangles to which a given edge belongs, divided by the number of triangles that might potentially include it, given the degrees of the adjacent nodes. Higher the edge's clustering coefficient, more is the likelihood that it is an intra-community tie. }

\section{Trussness based Hierarchy and Nature of Ties}
\label{truss-sec} 
The concept of k-truss was proposed by Cohen as a method to hierarchically decompose a graph into subgraphs with specific properties \citep{Cohen08}. Our motivation for using k-truss decomposition is to ascertain the hierarchical level of nodes in order to gauge their importance. We also exploit the decomposition to approximate community structure and to infer the strength of ties. The formal definition of k-truss adapted from \citep{Cohen08} follows. 

\begin{definition}
\uline{k-truss}:  The subgraph $T_k$ of graph $G$, where $k \geq 2$, is the k-truss of G,  iff each edge in $T_k$ is a part of at least $(k-2)$ closed triads, and $T_k$ is the maximal graph with this property. 
\end{definition}

\begin{figure}[!htbp]
%\centering
\includegraphics[height=1.8in,width=5.0in]{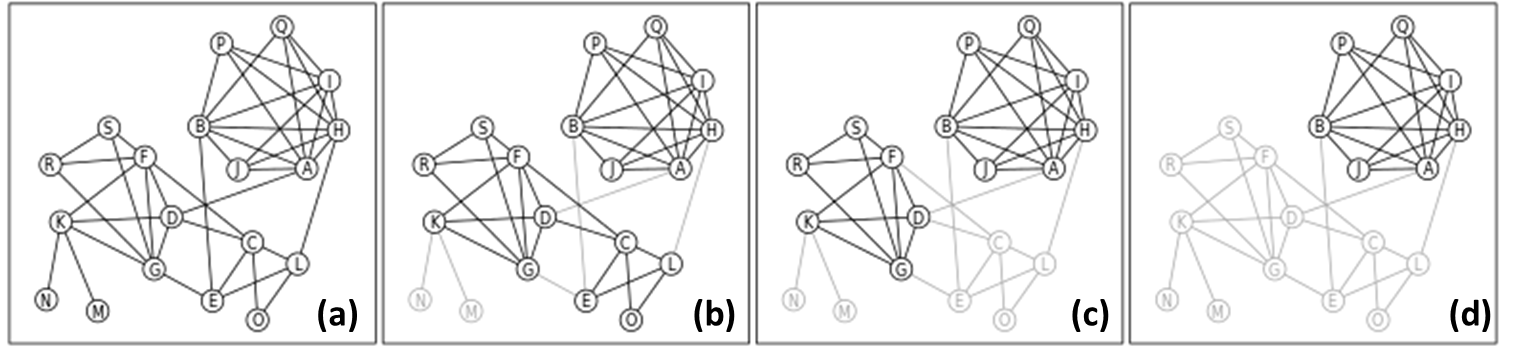} %trussexample-trussSCT.png}
\caption{Truss Decomposition  example (a) Toy Network (2-truss) with n=19 and m=42, (b) 3-truss subgraph, (c) 4-truss subgraph with two components, (d) 5-truss subgraph.}
\label{trussdecomposition}
\end{figure}  
%\vspace{-1em}
By definition, 2-truss is the graph $G$ itself. k-truss of $G$ is computed by iteratively removing edges that are not part of (k-2) triangles, until no more edges can be removed. 
\begin{example}
\textit{Fig. \ref{trussdecomposition}(a-d) shows k-truss decomposition of a toy network. The sequence of  images show  the  hierarchy of  subgraphs representing increasingly denser regions. Edges not contained in k-truss are dimmed for increasing values of k. }
\end{example}
We adapt the definition of  edge trussness from \citep{wang-vldb12}.
\begin{definition}
\label{edge-truss}
\uline{Edge Trussness} $t_{ij}$ of $e_{ij} \in E$  has value $k$,  iff $e_{ij} \in T_k\wedge  e_{ij}\notin T_{k+1}$. 
\end{definition} 
\begin{example}
\textit{Fig. \ref{trussexample-trussness} illustrates edge trussness of toy network  shown in Fig. \ref{trussdecomposition}(a). Edge {\it $e_{EG}$} has trussness 2 because it is present in 2-truss but absent in 3-truss. Edge $e_{FG}$ has trussness 4, since it belongs to 4-truss but not to 5-truss.}
\end{example}
\begin{figure}[!htbp]
\begin{minipage}[c]{0.60\textwidth}
\includegraphics[height=1.6in,width=2.3in]{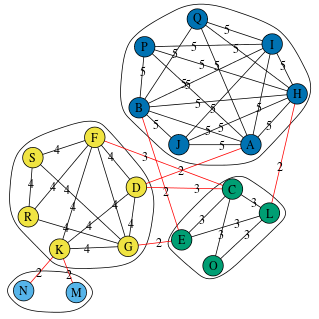}
\end{minipage}
\begin{minipage}[c]{0.35\textwidth}
\caption{\label{trussexample-trussness}Vertices with same node trussness bear the same colour and are grouped inside a polygon. Edges are labelled with their respective  trussness. Intra--community edges are coloured black and inter--community edges are coloured red.}
\end{minipage}
\end{figure} 
%\vspace{-1em}
Trussness, by definition, is a property of an edge. It \textit{indicates} the strength of tie between two nodes based on the number of common neighbours.  In pursuit of the ultimate goal of determining hierarchy-level of a node as well as for discriminating intra- and inter-community ties, we define trussness of a node as the maximum trussness of its incident edges. 
\begin{definition}
\label{node-trussnss}
\uline{Node Trussness} $\tau_i$ of node $v_i \in V$  is the maximum of trussness of edges incident on it i.e.
$\tau_i = \max\limits_j (t_{ij})  $  %, j = 1, \ldots n  
\end{definition}
By definition, an edge with trussness $k$ belongs to k-truss but not to (k+1)-truss. Proposition \ref{prop1} proves that the same property holds for a node with trussness $k$. 
\begin{Proposition}
\label{prop1}
Given a node $v_i$ with trussness $\tau_i = k$, $v_i \in T_k \wedge  v_i\notin T_{k+1}$
\end{Proposition}
\begin{proof}
We prove this by contradiction. Given that $\tau_i = k$, at least one edge incident on $v_i$ has trussness $k$ (by Def. \ref{node-trussnss}). Consequently, $v_i \in T_k$. Now suppose that $v_i \in T_{k+1}$. This implies there exists at least one edge incident on $v_i$ with trussness $k+1$. Then, by Def. \ref{node-trussnss}, $v_i$ must have trussness $k+1$, which contradicts the assumption that $\tau_i = k$.  Hence $v_i \notin T_{k+1}$.  $\hfill \square$
\end{proof}

Proposition \ref{prop1} implies that \textit{all} vertices with node trussness $k$ belong to the k-truss; the maximal truss number to which they can belong. This permits us to use node trussness as property to define hierarchy levels in the network. Given a graph decomposition, each hierarchy level in the graph is composed of all vertices with the same node trussness.
\begin{example}
\textit{In Fig. \ref{trussexample-trussness} vertices with the same node trussness are demarcated with same colour. Node trussness of vertex $A$ is 5, the maximum trussness of all its incident edges. Vertices $E$, $C$, $L$, and $O$ belong to the same hierarchy level since their node trussness is 3.}
\end{example}
Two nodes $v_i$ and $v_j$ with the same trussness $\tau_i = \tau_j = k$ are at the same hierarchy level. We posit that if they are linked by an edge $e_{ij}$ with trussness $t_{ij}=k$, they are part of  the same community, and the connecting edge is an intra-community edge. Accordingly, we define Trussness Matrix as below.
\begin{definition}
\uline{Trussness Matrix} $\Theta_{n \times n}$ is a symmetric boolean matrix  denoting the nature of ties (intra-- and inter--community) of $G(V,E,W)$. Element $\theta_{ij}$ of $\Theta$ is defined as 
\begin{equation}
    \theta_{ij}=
    \begin{cases}
      1, & \text{if}\ e_{ij} \in E \wedge (\tau_i = \tau_j = t_{ij}) \\
      0, & \text{otherwise}
    \end{cases}
\end{equation}
\end{definition}
A 1 in $\Theta$ denotes an intra--community link between pair of nodes, while the presence of 0 denotes either no link or inter--community link. Note that node and edge trussness have been used to approximate nature of links as intra- and inter-community ties.  The approximation is based on the sociological theory that the number of common neighbours shared by the endpoints of an edge account for the distinction between intra- and inter-community edges \citep{Nick2013}.

%Following Proposition forms the basis of the approximation.
%\begin{Proposition}
%\label{prop2}
%Given vertices $v_i, v_j$ with trussness $\tau_i = \tau_j = k$ connected by edge $e_{ij}$ with trussness $t_{ij} = k$. For all %vertices $v_x \in N_i$ and $v_y \in N_j$, $\vert N_i \cap N_j \vert \geq \vert N_i \cap N_x \vert$ and $\vert N_i \cap N_j \vert %\geq \vert N_j \cap N_y \vert$. 
%\end{Proposition}
%\begin{proof}
%Given that $t_{ij} = k$, by Def. \ref{edge-truss} edge $e_{ij}$ is part of at least ($k-2$) triangles. Thus, $v_i$ and $v_j$ have at least ($k-2$) common neighbours, i.e., $\vert N_i \cap N_j \vert \geq (k-2)$. Given $\tau_i = k$, by Def. \ref{node-trussnss}, $t_{ix} \leq  k \mbox{  } \forall v_x \in N_i$. This implies $\vert N_i \cap N_x \vert \leq (k-2)$. Hence $\vert N_i \cap N_j \vert \geq \vert N_i \cap N_x \vert$. By similar reasoning  $\vert N_ i \cap N_j \vert \geq \vert N_j \cap N_y \vert$. $%\hfill \square$
%\end{proof}
%By Proposition \ref{prop2}, two adjacent vertices $v_i$ and $v_j$ with the same trussness as that of the connecting edge share more common neighbours with each other than with any of their other neighbours. Hence, we conjecture that these nodes belong to the same community and the connecting edge is an intra-community link.

\begin{example}
\textit{In Fig. \ref{trussexample-trussness}, edge {\it $e_{AB}$} is a intra-community tie (coloured black) because $\tau_A = \tau_B = t_{AB} = 5$ whereas edge {\it $e_{AD}$} is a inter-community tie (coloured red).}
\end{example}
\section{Social Capital based Centrality Measure}
\label{sc-sec}
Formally, {\it Centrality} is a real-valued function $C: V \rightarrow \mathbb{R}^+$, which assigns a positive real number to each node indicating its importance. The proposed \textit{Social Centrality} (SC) measure designed for social networks takes into consideration the potential to access resources in contrast to existing measures that take into account actual resources held by a node. The design is intuitive and agrees with real-life experiences in human society. SC considers both direct ties as well as embeddedness of an actor to measure potential gains through the intensity of personal contacts, positional hierarchy and community structure.
\subsection{Sociability Index}
Sociability captures personal relationships people develop with each other through a history of interactions, consequent to commonalities, solidarity and trust. It quantifies the extent to which an actor can leverage the resources controlled by its immediate neighbours at dyadic level. We define \textit{Sociability Index} to quantify sociability by taking into consideration size of the immediate neighbourhood (ego-network) and intensity of relationships (edge-weights).
\begin{definition}
{\uline{Sociability index}} $\omega_i$ of node  $v_i$ is defined as the sum of weights of edges incident on it. 
\begin{equation}
\omega_i = \sum_{j}w_{ij}
\label{eqn-sociability}
\end{equation}
\end{definition}

Thus $\omega_i$ quantifies propensity to socialize by aggregating the intensity of direct relations of a node, indicating the potential gain from its ego-net. It determines the extent to which the node can mobilize resources (social capital) in the network. %Recall that  $w_{ij}=0$ if $e_{ij} \notin E$. 
\subsection{Bonding Potential}
Embeddedness of an actor in the network community structure determines his bonding potential.  An individual located in a dense region of $G$, snug in the community,  has higher ability to access resources by virtue of his location. This ability is also impacted by sociability of the intra-community neighbours as well node's position in hierarchy quantified by node trussness. Effectively, bonding potential of an individual is determined by his hierarchical position in $G$ as well as by the sociability of his intra-community neighbours.

Based on the social nature of humans, we hypothesise that each new node that joins the network and is yet to link with an existing actor has the potential to bond. Let $\alpha_i$ denote innate potential of $v_i$ to bond. In practice, innate bonding potential can be derived based on meta-data available for nodes. For example, in collaboration networks, $\alpha_i$  can be computed from the number of researchers belonging to the same organization.

Bonding potential $\beta_i$ of $v_i$ is defined as sum of  sociability index of its intra-community neighbours weighted by its position in hierarchy.
\begin{definition} 
{\uline{Bonding potential}} $\beta_i$ of node $v_i$ is defined as  
\begin{equation}
\beta_i = \alpha_i +  \sum_{j}{\theta_{ij}\cdot \omega_{j}\cdot \tau_j}
\label{eqn-bonding}
\end{equation}
\end{definition}
Recall that $\theta_{ij}$ is 1 for neighbours within the community, and 0 otherwise.  Accordingly, only intra--community edges of a node contribute to its bonding potential. 
\subsection{Bridging Potential}
An actor with links in diverse communities can draw advantages that are not available within his community. The ability of an actor to bring in new resources, hitherto not available to the community, asserts his importance because of the onward transmission of advantages to other members of the community. E.g., pacts signed by the heads of two nations bring indirect advantages to the citizens who are touched by the pact.

It is realistic to assume that an individual is attracted to new vistas and has potential to bridge to new social groups in $G$. We denote the innate bridging potential of $v_i$ by $\delta_i$. Like $\alpha_i$, innate bridging potential can be determined from the meta-data. For example, in collaboration networks, $\delta_i$ can be quantified as the number of institutions in which a researcher has studied or worked.

The bridging potential of $v_i$ is the sum of the intensity of relationship with its inter-community neighbours scaled by their respective positions in the hierarchy. 
\begin{definition} 
{\uline{Bridging potential}} $\gamma_i$ of node $v_i$ is defined as  
\begin{equation}
\gamma_i = \delta_i +  \sum_{j}{\bar \theta_{ij}\cdot w_{ij} \cdot \tau_j}
\label{eqn-bridging}
\end{equation}
\end{definition}
Here $\bar \theta_{ij}$ is the complement of  $\theta_{ij}$ and has value 1 in case of an inter--community tie or no tie. Note that in case of no tie $w_{ij} = 0$. Thus only inter--community edges of a node contribute to its bridging potential.

\subsection{Social Centrality Score}
Social centrality score of a node quantifies its ability to mobilize resources in the network based on its location in the hierarchy, embeddedness in community and intensity of relationships with neighbours. A node with a high score is likely to be in the mainstream of resource flow within $G$.  On the other hand, a node with low score has relatively low ability to acquire resources and hence has low importance.

SC score $\Psi: V \rightarrow \mathbb R^+$ is a real-valued function that aggregates its sociability index, bonding and bridging potential.
\begin{equation}
\label{sc-eq}
\Psi_i = \mathcal{F}(\omega_i,  \beta_i, \gamma_i)
\end{equation}

Each attribute quantifies a node property in accordance with the tenets of the theory of social capital. Attribute $\omega_i$ denotes the intensity of interactions with the immediate neighbours. Bonding ($\beta_i$) and bridging ($\gamma_i$) potentials are computed based on the embeddedness of the node in the community structure while taking cognizance of the network hierarchy.

Choice of aggregator function $\mathcal{F}$ to compute the score depends on the type of social network being analysed. E.g., in organizational networks, bonding potential is more important to achieve the respective goals of the project teams. To analyse such networks, the aggregator function can be designed to assign a higher weight to $\beta_i$. In collaboration networks,  bridging potential is more important since multi-disciplinary researchers bring-in innovation in the group. Accordingly  $\mathcal{F}$ can be designed to accord higher importance to $\gamma_i$ in research networks. 

For simplicity and efficiency, we choose a  multiplicative function as the aggregator ($\mathcal{F}$) in our experiments. We also  set $\alpha_i=\delta_i=1 \; \forall v_i \in V$ to focus on the facets of social capital. Thus we compute SC score of node $v_i$ as 
\begin{equation}
\Psi_i = \omega_i \cdot  (1+\beta_i) \cdot (1+\gamma_i) 
\end{equation}

Design of context-specific aggregator functions for different types of social networks is an area of intense study and beyond the current scope. 
\subsection{Algorithmic Complexity}
SC algorithm uses $O(m+n)$ space to hold the input graph, $O(m)$ space for holding trussness of edges and $O(n)$ space for holding trussness as well as bonding and bridging potential of nodes. For k-truss decomposition, we use an elegant in-memory k-truss decomposition algorithm with time complexity $O(m^{1.5})$ as proposed by  Wang et al. \citep{wang-vldb12}. Computation of edge trussness thus takes $O(m^{1.5})$ time. Computation of node trussness iterates over edges and hence takes $O(m)$ time. Subsequent computation of sociability index, bridging and bonding potential again iterates over edges and takes $O(m)$ time. Thus algorithmic space complexity is $O(m+n)$ and time complexity is $O(m^{1.5})$. In future work, we intend to improve the edge trussness computation time by using more efficient triangle listing algorithms proposed in \citep{Ortmann2014}. 
\section{Experimental Evaluation}
\label{expt-sec}
The goal of experimental evaluation is to assess the performance of proposed Social Centrality measure (SC). The experimental study is designed to examine the following specific questions. 
\begin{enumerate}[(i)]
\item \textit{Is the social importance captured by SC measure consonant with known facts about the social importance of individuals in a social network? }

We examine this question in Sec. \ref{sec-validation} using a well-researched small terrorist network. This investigation establishes the validity of the proposed measure and permits reasoning about performance of SC measure.
\item \textit{Is the proposed measure effective for  large real-world social networks?} \\ 
To investigate this question,  we use four large networks that sport ground truth. The first is an e-mail network and other three are collaboration networks. We compare top ranked actors delivered by various centrality measures in this examination. Depending on the size of the network, we use either top-100 or top-500 actors for comparison with top-10 actors as per the ground truth available for the network (Sec. \ref{sec-effectiveness}). Later we drill down to check the overlap of predicted top rankers with ground truth (Sec. \ref{sec-quality}). 

\item \textit{How close are the ranks delivered by SC to ground truth in large networks?} 

This investigation studies correlation of ground truth ranking of all actors and their ranking delivered by compared centrality measures (Sec. \ref{sec-consonance}). 
\item \textit{To what extent does the SC ranking degrade because of using truss approximation of communities? What is the speed advantage of this approximation compared to executing community detection algorithm?}

We compare the performance of SC with another version SC-Com that uses a community detection algorithm \citep{blonde08}\footnote{We experimented with various community detection algorithms available in the Python \textit{igraph} library, and chose the Multilevel algorithm because it was the fastest.}. In this version, analogous to the trussness matrix $\Theta_{n \times n} $, the community matrix ($\Pi_{n \times n}$) is created where $\pi_{ij} = 1$ if $v_i$ and $v_j$ are in the same community and $0$ otherwise. Secs. \ref{sec-consonance}, \ref{sec-scalability} describe the results of this experiment.
\item \textit{Is the SC measure scalable w.r.t. large networks?}

We examine scalability of SC using several real-world and synthetic large networks (Sec. \ref{sec-scalability}).
\end{enumerate}
%\subsection{Baseline Measures}

We evaluated performance of SC measure in comparison to classical centrality measures - Degree Centrality (DC), Eigenvector Centrality (EC), Betweenness Centrality (BC) and Closeness Centrality (CC),  and recently proposed   measures - Spanning Tree Centrality (STC), Laplacian Centrality (LC) and Comm Centrality (CoC). Further, we study the effectiveness of SC compared to two recently proposed measures of social capital - Network Constraint (NC) and SoCap (Sec. \ref{sec-related}). We also compare the performance of SC with SC-Com.

We implemented algorithm to compute SC in {\it Python} (32bits, v 2.7.3) and executed on Intel Core i7-6700 CPU @3.40GHz with 8GB RAM, running UBUNTU 16.04{\footnote{Weighted versions of classical centrality measures were used from the Python \textit{igraph} library; Python code for implemented measures and datasets used for experimentation are available at https://github.com/rakhisaxena/SocialCentrality}}. 
We report the results of experimentation in the following sub-sections.  
\subsection{Preliminary Investigation}
\label{sec-validation}
For the preliminary investigation, we use the Bali Terrorist Network \citep{Koschade2006}. This well studied weighted network with 17 actors and 63 edges represents communication between terrorists from the Jemaah Islamiyah cell responsible for Bali bombings in 2002 (Fig. \ref{bali-network}). Edges between actors are weighted by the strength of their relationship in range of (1-5), with 1 signifying the weakest and 5, the strongest relationship. 
\begin{figure}[!htbp]
\begin{minipage}[c]{0.60\textwidth}
\includegraphics[height=2.4in,width=2.5in]{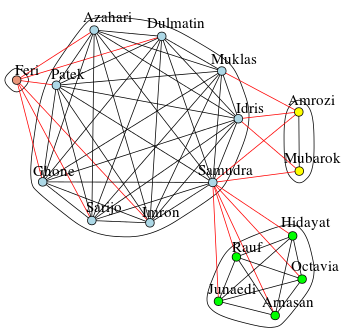}
\end{minipage}
\begin{minipage}[c]{0.35\textwidth}
\caption{\label{bali-network} Bali Terrorist Network -- Nodes with the same trussness are grouped inside a polygon and are marked with identical colour. Edges connecting nodes with the same trussness are coloured black and edges connecting nodes with dissimilar trussness are coloured red.}
\end{minipage}
\end{figure} 
%\vspace{-1em}
\input{tab-baliscores.tex} 

Table \ref{tab-baliscores} lists the ranks assigned by different centrality measures. It is documented in \citep{Laplacian13} that terrorist Samudra was responsible for the actual bombing strategy and was the only contact between the bomb makers and the group setting off the bombs. He is ranked highest by all measures except SC-Com and CoC. This is because the community detection algorithm used \citep{blonde08} places Samudra with the smaller group of nodes (marked in green) thus reducing his bonding potential.

Idris, the logistics commander and Imron, the team's gofer, both holding key roles in the operation \citep{Qi2015} are ranked second and third by SC and majority of other measures. According to \citep{Koschade2006}, Feri, one of the suicide bombers had feeble ties with other members of the team and hence socially less important. Only  SC, CC and SoCap measures are able to capture this fact and assign the lowest rank to Feri. 

\textit{For this dataset, all centrality measures generally agree with each other. SC ranks are consistent with the majority of other centrality measures, with relatively stronger agreements with CC and social capital based measures SoCap.}
\subsection{Effectiveness Analysis}
\label{sec-effectiveness}
We examine the effectiveness of SC on four real-world social networks -- an email network and three co-author networks. We find these networks suitable for the experiment because they are reasonably large, publicly available and sport ground truth. 
In the email network, employee's position in the organizational hierarchy and payment received reveal their social importance and hence are treated as ground truth. In co-author networks, citation count of an author serves as ground truth as it reflects author prestige and importance in the academic community. Table 2 shows structural features of the networks. Details of these networks and analysis of experimental results follow. 
\input{tab-datadesc.tex}

{\textbf{(i) Enron email network:}} We downloaded\footnote{http://www.kddcup2012.org/wcukierski/enron-email-dataset} Enron email dataset that contains $\approx$  500,000 emails generated by employees of the Enron corporation. The email network has senders and recipients as nodes and an edge indicates the exchange of email between the endpoints.  Edge-weight corresponding to an email sent to $n$ recipients is set as $1/n$. Multiple edges between each sender-recipient pair are coalesced by summing the edge-weights.

Table \ref{tab-top10Enron} presents ranks assigned by each measure to top-10 employees in the Enron organization hierarchy. STC and SoCap could not complete execution on our machine\footnote{STC ran out of memory and SoCap did not complete even after running for a day.}. The first three columns of Table \ref{tab-top10Enron} list top-10 employees, their organizational position and respective payment they received\footnote{Enron employee position in organization structure obtained from http://www.infosys.tuwien.ac.at/staff/dschall/email/enron-employees.txt.}. Rest of the columns show the ranks assigned by different measures in their respective top-100 list. It is evident from the table that only SC is able to find all top-10 employees in its top-100 list. Other measures discover only a few, with NC finding the least. 

The last row of Table \ref{tab-top10Enron} shows the rounded off Root Mean Square Error (RMSE) between actual ground truth ranks of top-10 employees and ranks assigned by respective measure. Let $R^m$ be the vector of ranks assigned by measure $m$ to top-k actors. We compute the $RMSE^m$ of measure $m$ as 
\begin{equation}
\label{eqn-rmse}
\scriptsize{RMSE^m = \sqrt{\frac{1}{k}\sum_{t=1}^{k}(R^m[t] - t)^2}}
\end{equation}

SC supports the least RMSE indicating its superior ability to deliver ranks in consonance with ground truth. 

\textit{Overall, centrality scores delivered by  SC are closer to known facts about the network. This ability is attained because of the due consideration that is given to hierarchical positions and interaction patterns. }
\input{tab-top10Enron.tex}

{\textbf{(ii) HepTh co-author network:}} We downloaded HepTh\footnote{https://www.cs.cornell.edu/projects/kddcup/datasets.html}  dataset which consists of (i) paper to author mapping of the High Energy Physics Theory portion of the arXiv until May 1, 2003 ($\approx$ 29K papers), (ii) citing-cited relation between papers. The co-author network consists of nodes representing authors and edge between authors indicates co-authorship on the same paper. Weight of edge representing co-authorship on paper with \textit{n} authors (including self) is set as $1/(n-1)$. Multiple edges are collapsed by summing edge-weights. 
\input{tab-top10HepTh.tex}

Table \ref{tab-top10HepTh} shows top-10 cited authors in this network and ranks assigned by various measures in the expanded list of  top-100 ranks. STC algorithm could not complete execution on this (smallest of 4) network while EC performs worst (3/10). SC retrieves the maximum (8/10) authors with least RMSE value, demonstrating its effectiveness in capturing the social importance of authors in collaboration network. NC retrieves fewer authors with apparently better ranks but with higher RMSE value. CC and SoCap, which were close competitors in Bali network perform poorly in this network. We revisit and analyse the ranks in Sec. \ref{sec-consonance}.
\input{tab-top10DBLP.tex}

{\textbf{(iii) DBLP co-author network:}} The co-author dataset\footnote{https://aminer.org/DBLP\_Citation} and its related citation dataset was downloaded from DBLP (1,632,442 papers and 2,327,450 citation relationships till year 2011). Weights in this co-author network are assigned in the same manner as done for HepTh network. 

We report results for seven measures - SC, SC-Com, DC, EC, LC, NC and CoC in Table \ref{tab-top10DBLP}. During execution, STC ran out of memory, and BC, CC, and SoCap could not complete in 24 hours on our machine. Table \ref{tab-top10DBLP} lists top-10 cited authors in the DBLP network along with their ranks assigned by different measures. The absence of the author in top-500 ranks is denoted by `-'. 

SC was able to retrieve maximum authors (7/10) in top-500 ranks with least RMSE value, establishing its effectiveness. EC fails to find any author in the top-10 cited authors confirming that it is not an appropriate centrality measure for collaboration networks. Rest of the measures  SC-Com, DC, LC, NC, and CoC retrieve fewer authors than SC and exhibit larger RMSE value.  Ranks of DBLP dataset are analysed further in Sec. \ref{sec-consonance}.

{\textbf{(iv) USPTO co-author network:}} The downloaded\footnote{http://www.nber.org/patents/} USPTO data consists of inventor details including co-authorship of all patents granted from 1975-1999 by the US Patent Office (USPTO) and citing-cited relation between patents. 

We report results for seven measures SC, SC-Com, DC, EC, LC, NC, and CoC in Table \ref{tab-top10USPTO}. For this network too, STC ran out of memory, and BC, CC, and SoCap could not complete in 24 hours on our machine. Table \ref{tab-top10USPTO} lists top-10 cited inventors in the USPTO network along with their ranks assigned by different measures.

SC was able to find 5/10 inventors in top-500 while other measures are able to retrieve much fewer. SC is unable to find George Spector in its top-500 rankers but DC, LC, and CoC assign him rank 1. While investigating the failure of SC to retrieve this inventor, we discovered that George Spector is apparently not an inventor at all\footnote{http://www.entrepreneur.com/article/185660}. He runs a New York business that helps small-time inventors obtain patents. Spector adds his name to those inventions netting him high degree as well as high citations. Since his collaborations are with one-time inventors, his social capital is much less. Shunpei Yamazaki who is the most prolific inventor in 2005 as reported by USA Today\footnote{http://usatoday30.usatoday.com/money/industries/technology/maney/2005-12-13-patent\_x.htm} is assigned much higher rank by SC.  
\input{tab-top10Patent.tex}

\textit{Relatively strong agreement of SC score with ground truth in email and collaboration networks asserts its applicability to diverse contexts. Capitalizing on the underlying hierarchical and community structures prominent in social networks, SC effectively utilizes the strength of ties between actors to rank nodes.}
\subsection{Quality Analysis}
\label{sec-quality}
In view of available ground truth (citations) for all actors of DBLP, USPTO and HepTh networks, we check the quality of SC ranks using two metrics. We compare the top-cited authors with the top-k rankers predicted by compared measures and quantify the overlap using Jacquard Index. Admitting top-k authors delivered by a centrality measure as the retrieved set and top-k cited authors as the relevant set, we also evaluate classical information retrieval measures, precision, and recall. 
\subsubsection{Overlap between Top-cited Authors and Top-k Rankers}
Let $G_k$ be the set of top-k authors as per ground truth and $P^m_k$ be the set of top-k authors predicted by measure $m$. We compute the Jacquard Index ($JI^m$) of measure $m$ as follows.
\begin{equation}
\label{eqn-jacquard}
\scriptsize JI^m = \frac{|G_k \cap P^m_k|}{|G_k \cup P^m_k|} 
\end{equation}
We vary \textit{k} from 100 to 1000 and plot Jacquard Index for investigated centrality measures (Fig. \ref{fig-quality}(a),(b),(c)). The figures demonstrate that both the SC and SC-Com measures outperform the rest. SC retrieves more top-cited authors than SC-Com for DBLP and HepTh networks for lower values of \textit{k}, but for larger \textit{k}, SC-Com performs better than SC. For the USPTO network, the superior performance of SC compared to SC-Com is attributed to the fact that this network is extremely disconnected \citep{SoCap2014}. This is because, in the patent world, inventors tend to work in closed groups and patents across organizations/groups are rare.

EC identifies the least number of top-cited authors as it focuses only on the quality of its connections. It ignores the importance of actors arising from bridging ties. DC is the closest competitor of SC, with SC faring better with increasing $k$ (SC discovers more top-cited authors). Rest of the measures clearly lag behind SC in detecting important authors. This demonstrates that nodes with high social capital are high performers in the network because of the realistic factors that are taken into account by SC measure. 

\textit{Despite its relatively better performance, SC is unable to retrieve \textit{all} top-cited authors in all three co-author networks. Many external factors, such as working in many areas of research, quality of work, long publishing tenure, high productivity, etc. can potentially contribute to high citation count of an author. Even though collaboration with more influential authors correlates positively with higher citation count, it is not the only causal factor.}

\begin{figure}[!htbp]
\includegraphics[height=5.5in,width=5.5in]{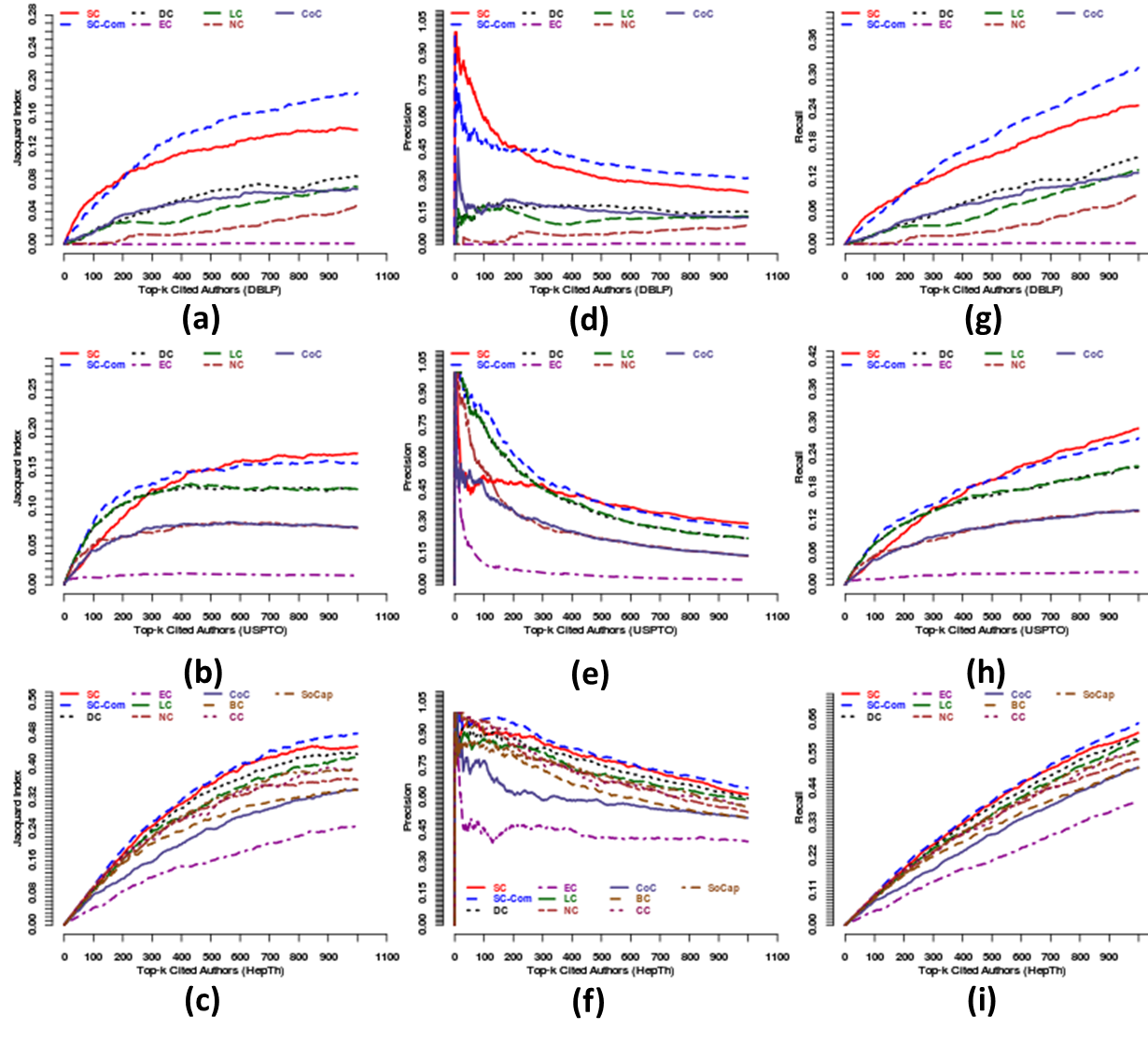}
\caption{Quality analysis of top-k cited authors (k = 1 to 1000) and consonance ranking of top-k authors retrieved by various measures using DBLP, USPTO and HepTh networks.}
\label{fig-quality} 
\end{figure} 
%\vspace{-2em}
\subsubsection{Precision/ Recall Evaluation}
For this experiment, the top-k cited authors are treated as the \textit{relevant} set of results and the top-k authors predicted by a centrality measure as the \textit{retrieved} set. Precision measures how many top-cited authors are identified by the centrality measure and recall measures the completeness of the retrieved results. We compute precision and recall measures as follows:
\begin{equation}
\scriptsize Precision = \frac{\vert relevant \cap retrieved\vert}{\vert retrieved \vert} \,\,\,\,\,\,\,\,
Recall = \frac{\vert relevant \cap retrieved \vert}{\vert relevant \vert}
\end{equation}
We show the precision for DBLP, USPTO, and HepTh networks in Fig. \ref{fig-quality}(d),(e) and (f) by varying the top-1000 authors retrieved by each method. Fig. \ref{fig-quality}(g),(h) and (i) show the recall for the examined datasets. For all three networks, both methods SC and SC-Com, outperform all other comparative measures in terms of precision and recall. For the DBLP network, SC-Com and SC are able to find upto $28\%$ and $24\%$ top-1000 most cited authors respectively. In case of the USPTO network, SC-Com and SC are able to find upto $25\%$ and $28\%$ top-1000 most cited authors respectively. As before, the superior performance of SC in the USPTO network is due to the disconnected nature of the network. SC-Com is able to find $62\%$ and SC $61\%$ of top-1000 most cited authors in case of the HepTh network. 
\input{tab-spearman.tex}
\subsection{Consonance of Ranks with Ground Truth}
\label{sec-consonance}
Next, we investigate ranks delivered by centrality measures against the actual citation ranks of the corresponding authors for collaboration networks. Table 7 shows Spearman's rank correlation coefficient ($\rho$) between ground truth ranks and those delivered by SC and competing centrality measures. SC performs better than all other centrality measures except SC-Com. This is expected because SC-Com uses exact community detection whereas SC performs truss-based approximation of community structure. In the next section, we show the speed advantage of using SC compared to SC-Com.

\subsection{Scalability Evaluation w.r.t. Large Networks}
\label{sec-scalability}
We now investigate the scalability of SC measure. Three synthetic networks viz. Erd\"{o}s-R\'{e}nyi (ER), Watts-Strogatz (WS) and Forest Fire (FF) models were used for analysis\footnote{Synthetic networks generated using \textit{igraph} package; ER: erdos.renyi.game(n, m=2n); WS: watts.strogatz.game(n, dim=1, nei=4, p=0.3); FF: forest.fire.game(n, ambs=4, bw.factor=0.2, fw.prob=0.3)}. These networks have the number of nodes varying from 1 million to 10 million and edges ranging from 2 million to $\approx$60 million.  We computed SC measure on each network and averaged running time over three runs. We found that the computation time of SC for each network individually was maximum 92 seconds (Fig. \ref{fig-runtime}).

Next, we examine the improvement in running time of computing SC measure over SC-Com. We computed SC and SC-Com measures for five publicly available real-world large networks\footnote{Real-world networks downloaded from http://www.snap.edu/data}and three synthetic networks. Table \ref{tab-runtime} shows the running time (in seconds) of both measures. SC runs much faster than SC-Com for all examined datasets vindicating the claim of its scalability.
\begin{figure}[!htbp]
\begin{minipage}[c]{0.40\textwidth}
\includegraphics[height=1.4in,width=1.8in]{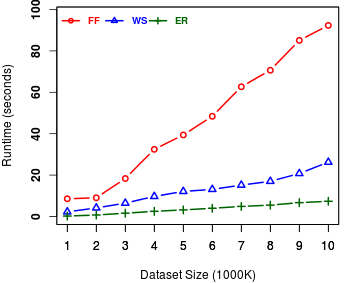}
\end{minipage}
\begin{minipage}[c]{0.6\textwidth}
\caption{\label{fig-runtime} Execution time of \textit{SC} for synthetic networks.}
\end{minipage}
\end{figure} 
\input{tab-scalability.tex}

\section{Conclusion}
Existing graph theoretic {\it centrality} measures take no cognizance of inherent hierarchy and community structure in social networks. These factors, however, both intuitively and theoretically have a compelling impact on the importance/power of an actor. Limited applicability of existing measures in social networks motivates proposal for \textit{Social Centrality} measure (SC).  

Based on the mature theory of {\it Social Capital}, SC emulates real-life behaviour of social actors to bond within community and bridge between communities. k-truss decomposition is performed to elicit network hierarchy and approximate community structure. Computation of SC is scalable because of underlying $O(m^{1.5})$ algorithm for k-truss decomposition.  
 
The empirical study based on diverse, real-life social networks vindicates the propositional basis of the model and demonstrates its effectiveness, superior performance, and scalability compared to prevailing centrality measures. 
\bibliographystyle{aps-nameyear}
\bibliography{sc-bib}
\end{document}

%% file: tab-baliscores.tex
% Please add the following required packages to your document preamble:
% \usepackage{booktabs}
\begin{table}[!htbp]
\centering \scriptsize 
\caption{\scriptsize Ranks of actors in the Bali terrorist network based on various centrality measures.}
\label{tab-baliscores}
\begin{tabular}{@{}llllllllllll@{}}
\toprule
\textbf{Actor}    & \textbf{SC} & \textbf{SC-Com} & \textbf{DC} & \textbf{EC} & \textbf{BC} & \textbf{CC} & \textbf{STC} & \textbf{LC} & \textbf{NC} & \textbf{SoCap} & \textbf{CoC}\\ \midrule
\textbf{Samudra}  & 1    & 2       & 1            & 1           & 1           & 1           & 1            & 1           & 1           & 1    & 3          \\
\textbf{Idris}    & 2   & 1         & 2            & 8           & 2           & 3           & 2            & 2           & 2           & 2       & 1       \\
\textbf{Imron}    & 3    & 3        & 3            & 3           & 3           & 2           & 3            & 3           & 3           & 3          & 4    \\
\textbf{Dulmatin} & 4   & 4        & 3            & 3           & 6           & 4           & 9            & 3           & 5           & 5       & 4       \\
\textbf{Ghone}    & 4   & 4         & 3            & 7           & 5           & 4           & 4            & 3           & 5           & 5     & 4         \\
\textbf{Patek}    & 6   & 6         & 3            & 3           & 6           & 6           & 4            & 3           & 7           & 8     & 4         \\
\textbf{Sarijo}   & 6   & 6         & 3            & 6           & 6           & 6           & 4            & 3           & 7           & 7    & 4          \\
\textbf{Azahari}  & 8   & 8         & 3            & 2           & 6           & 8           & 4            & 3           & 9           & 4    & 4          \\
\textbf{Muklas}   & 9    & 9        & 3            & 9           & 6           & 9           & 8            & 9           & 4           & 10   & 2          \\
\textbf{Arnasan}  & 10     & 13      & 11           & 13          & 6           & 11          & 11           & 11          & 12          & 11   & 11          \\
\textbf{Rauf}     & 10   &13        & 11           & 17          & 6           & 11          & 11           & 11          & 12          & 12       & 11      \\
\textbf{Octavia}  & 10   & 13        & 11           & 13          & 6           & 11          & 11           & 11          & 12          & 12       & 11      \\
\textbf{Hidayat}  & 10  & 13         & 11           & 16          & 6           & 11          & 11           & 11          & 12          & 12       & 11      \\
\textbf{Junaedi}  & 10   & 13        & 11           & 13          & 6           & 11          & 11           & 11          & 12          & 12      & 11       \\
\textbf{Amrozi}   & 15      & 10     & 16           & 11          & 4           & 10          & 16           & 16          & 10          & 9       & 16       \\
\textbf{Mubarok}  & 16    & 11       & 17           & 12          & 6           & 16          & 17           & 17          & 17          & 16      & 17       \\
\textbf{Feri}     & 17     & 12      & 10           & 10          & 6           & 17          & 10           & 10          & 11          & 17    & 10         \\ \bottomrule
\end{tabular}
\end{table}
%\vspace{-1em}

%% file: tab-datadesc.tex
% Please add the following required packages to your document preamble:
% \usepackage{booktabs}
\begin{table} [!htbp]
\label{tab-datadesc}
\begin{minipage}{3.07in}
\scriptsize {
\begin{tabular}{@{}lrrllll@{}}
\toprule
\textbf{N/W} & \multicolumn{1}{c}{\textbf{n}} & \multicolumn{1}{c}{\textbf{m}} & $\mathbf{\overline  k}$ & $\mathbf{k_{max}}$ & \textbf{gcc} & \textbf{T} \\ \midrule
\textbf{HepTh}   & 11767      & 41323      & 7.02            & 330             & 0.23         & 10         \\
\textbf{Enron}   & 73318      & 283500     & 7.73            & 3038            & 0.07         & 36         \\
\textbf{DBLP}    & 971783     & 3481524    & 7.16            & 889             & 0.25         & 173        \\
\textbf{USPTO}   & 1159188    & 2719279    & 4.69            & 685             & 0.26         & 34          \\ \bottomrule
\end{tabular}
} 
\end{minipage}
\begin{minipage}[c]{0.35\textwidth}
\caption{\scriptsize Structural Properties of N/Ws. $\overline  k$-Avg. Degree, $k_{max}$-Max. Degree, $gcc$-Global Clustering Coefficient, $T$-Highest Truss Level.}
\end{minipage}
\end{table}
%\vspace{-3em}

%% file: tab-top10Enron.tex
% Please add the following required packages to your document preamble:
% \usepackage{booktabs}
\begin{table*}[!htbp]
\centering \scriptsize
\caption{\scriptsize Top-10 socially important employees (on the basis of their position and payment received in Millions ($P_M$)) in Enron email network and the ranks assigned to them by various measures. `-' denotes absence of employee in the top-100 list of the measure.}
\label{tab-top10Enron}
\begin{tabular}{@{}llllllllllll@{}}
\toprule
\textbf{Employee}        & \textbf{Position}                 & \textbf{$\mathbf{P_M}$} & \textbf{SC} &  \textbf{SC-Com}  & \textbf{DC} & \textbf{EC} & \textbf{LC} & \textbf{NC} & \textbf{BC} & \textbf{CC} & \textbf{CoC} \\ 
\midrule
\textbf{K. Lay}     & CEO                                 & \$103      & 26   & -         & 1             & -             & 1             & 17            & 14            & -  & 1           \\
\textbf{J. Lavorato}   & CEO                & \$10       & 15    & 9        & 21            & 29            & 79            & -             & 6             & 1   & 20          \\
\textbf{D. Delainey}  & CEO                      & \$4        & 18    & 8        & -             & 95            & 50            & -             & 15            & 5   & 84          \\
\textbf{M. Haedicke}   & MD                   & \$3        & 14  & 18          & 59            & -             & 10            & -             & 29            & 67  & -          \\
\textbf{L. Kitchen}  & President                           & \$3        & 13    & 10        & 11            & 39            & 41            & -             & 12            & 6   & 7          \\
\textbf{R. Buy}        & CRO      & \$2       & 33   & 31         & 43            & 76            & 45            & -             & -             & -  & 43           \\
\textbf{S. Kean}     & VP    & \$1        & 6     & 6        & 46            & -             & 38            & -             & 2             & 10  & 98          \\
\textbf{R. Shapiro} & VP & \$1        & 16     & 25       & 11            & -             & -             & -             & 78             & 9  & -           \\
\textbf{S. Beck}      & COO           & \$0.9         & 11    & 7         & 3             & 80            & 17            & -             & 23            & 57   & 11         \\
\textbf{J. Derrick}   & Lawyer                     & \$0.5          & 41     & -       & -             & -             & -             & -             & -             & -   & -          \\ 
\midrule
$\mathbf{\lceil RMSE \rceil}$ & & & 17 & 322	& 54 &	391 &	64 &	457 &	61 &	273 &	115 \\
 \bottomrule
\end{tabular}
\end{table*}
%\vspace{-1em}

%% file: tab-top10HepTh.tex
\begin{table}[!htbp]
\centering \scriptsize {
\caption{\scriptsize Top-10 cited authors in HepTh co-author network and the ranks assigned to them by various measures. `-' denotes the absence of the author in the top-100 list of the measure.}
\label{tab-top10HepTh}
\begin{tabular}{lllllllllll}
\toprule
\textbf{Name} & \textbf{SC} & \textbf{SC-Com} &\textbf{DC} & \textbf{EC} & \textbf{LC} & \textbf{NC} & \textbf{BC} & \textbf{CC} & \textbf{SoCap}& \textbf{CoC} \\ \midrule
\textbf{E. Witten}      & 29   & 32         & 97            & -             & -             & 4             & 15            & 2             & 37   & 55            \\
\textbf{N. Seiberg}     & -   & -          & -             & -             & -             & -             & -             & -             & -     & -           \\
\textbf{M. R. Douglas}  & 87   & 74         & -             & -             & -             & 2             & 21            & 6             & -     & -           \\
\textbf{C. Vafa}        & 20   & 22         & 42            & -             & 56            & 13            & 61            & 21            & 30  & -              \\
\textbf{S. S. Gubser}   & -   & 33          & -             & 37            & -             & -             & -             & -             & -   & -             \\
\textbf{N. Seiberg}     & 55   & -         & -             & -             & -             & 33            & -             & 37            & -                \\
\textbf{A. A. Tseytlin} & 18 & 7           & 20            & 38            & 32            & 5             & 1             & 95            & 5       & 2         \\
\textbf{I. R. Klebanov} & 60   & 98         & -             & -             & -             & -             & -             & -             & 85 & -              \\
\textbf{A. Strominger}  & 98   & 59         & 60            & -             & 93            & 8             & -             & 19            & 66    & 24           \\
\textbf{P. K. Townsend} & 61     & 13       & 27            & 19            & 24            & -             & 11            & 97            & 54    & -           \\ \midrule
$\mathbf{\lceil RMSE \rceil}$ & 103 & 270 & 124 & 498 & 144 & 162 & 453 & 353 & 329 & 277 \\
 \bottomrule
\end{tabular}
}
\end{table}

%% file: tab-top10DBLP.tex
\begin{table}[!htbp]
\centering \scriptsize {
\caption{\scriptsize Top-10 cited authors in DBLP co-author network and their ranks assigned by various measures. `-'  denotes absence of the author in the top-500 list of the measure.}
\label{tab-top10DBLP}
\begin{tabular}{@{}llllllll@{}}
\toprule
Name               & \textbf{SC} & \textbf{SC-Com} & \textbf{DC} & \textbf{EC} & \textbf{LC} & \textbf{NC} & \textbf{CoC} \\ 
\midrule
\textbf{Jeffrey D. Ullman}     & 18   & 215        & -          & -           & -           & -    &   -   \\
\textbf{Rakesh Agrawal}       & 21  & 130        & 467          & -           & -           & -    &  -     \\
\textbf{Hector Garcia-Molina} & 1   & 7         & 169           & -           & 174         & 423 & -       \\
\textbf{David S. Johnson}      & -  & -          & -            & -           & -           & -   & -       \\
\textbf{Jiawei Han}           & 99  & 10        & 70           & -           & 66          & 433   & 25     \\
\textbf{Scott Shenker}        & -    & -       & -          & -           & -           & -       &  -  \\
\textbf{Christos Faloutsos}   & 239  & 19        & 94           & -           & 97          & 178  &  219     \\
\textbf{David E. Culler}       & 181  & -        & 430          & -           & -           & -   &  -      \\
\textbf{David J. DeWitt}       & 25   & 54       & 353          & -           & -           & -    &   -    \\
\textbf{Hari Balakrishnan}    & -    & -        & -            & -           & -           & -    &   -    \\ \midrule
$\mathbf{\lceil RMSE \rceil}$ & 1057 & 1161 & 1494 & 19462 & 1368 & 6157 & 15094 \\ 
\bottomrule
\end{tabular}
}
\end{table}
%\vspace{-1em}

%% file: tab-top10Patent.tex
\begin{table}[!htbp]
\centering \scriptsize{
\caption{\scriptsize Top-10 cited authors in USPTO co-author network and their ranks assigned by various measures. `-' denotes absence of the inventor in the top-500 list of the measure.} 
\label{tab-top10USPTO}
\begin{tabular}{@{}llllllll@{}}
\toprule
Name       & \textbf{SC} & \textbf{SC-Com} & \textbf{DC} & \textbf{EC} & \textbf{LC} & \textbf{NC} & \textbf{CoC}\\ \midrule
\textbf{Felix Theeuwes}     & 260  &  -       & -          & -           & -           & -  & -         \\
\textbf{Donald E. Weder}       & 91 & -         & -          & -           & -           & -  & -        \\
\textbf{David T. Green } & 360  & -          & -          & -           & -         & -   &  -    \\
\textbf{Yasushi Sato}      & 160    & 83        & 159            & -           & 145           & 95 & -          \\
\textbf{Shunpei Yamazaki}           & 16   & 107       & 328           & -           & 376          & - & -       \\
\textbf{Jerome H. Lemelson}        & -   & -        & -          & -           & -           & -   &  -      \\
\textbf{Roshantha Chandra}   & -   &-       & -           & -           & -          & -  &  -     \\
\textbf{George Spector}       & -  & 252        & 1          & -           & 1           & 1  &  1       \\
\textbf{Yasushi Takatori}       & - & -         & -          & -           & -           & - &  -        \\
\textbf{Yoshiaki Shirato }    & -  & -          & -            & -           & -           & - & -          \\ \midrule
$\mathbf{\lceil RMSE \rceil}$ & 10565 & 19565 & 22214 & 305902 & 64403 & 84953 & 49124 \\
\bottomrule
\end{tabular}
}
\end{table}
%\vspace{-1em}

%% file: tab-spearman.tex
\begin{table}[!htbp]
\label{tab-spearmanrho}
\begin{minipage}{3.55in}
\scriptsize {
\begin{tabular}{llllllll}
\hline
\textbf{N/W} & \textbf{SC} & \textbf{SC-Com} & \textbf{DC} & \textbf{EC} & \textbf{LC} & \textbf{NC} & \textbf{CoC} \\ \hline
\textbf{HepTh}            & 0.65        & 0.70         & 0.64        & 0.54        & 0.64        & 0.58        & 0.62         \\ \textbf{DBLP}             & 0.45        & 0.50         & 0.42        & 0.38        & 0.43        & 0.39        & 0.42         \\ \textbf{USPTO}            & 0.38        & 0.36         & 0.33        & 0.18        & 0.29        & 0.32        & 0.33         \\ 
\hline
\end{tabular}
}
\end{minipage}
\begin{minipage}[c]{0.22\textwidth}
\caption{\scriptsize Spearman's $\rho$ between ground truth ranks and centrality measures.}
\end{minipage}
\end{table}
\vspace{-3em}

%% file: tab-scalability.tex
\begin{table} [!htbp]
\centering
 \scriptsize
 {
 \caption{ \scriptsize Runtime Comparison of \textit{SC} and \textit{SC-Com} Centrality measures.} 
 \label{tab-runtime}
 \begin{tabular}{lllcl}
\hline
\textbf{Network} & \textbf{n} & \textbf{m} & \textbf{Running Time} & \textbf{(in seconds)} \\ \cline{4-5}
& $\mathbf{(10^6)}$   & $\mathbf{(10^6)}$ & \textbf{SC}   & \textbf{SC-Com}    \\ \hline
\textbf{Amazon0601}                   & 0.40     & 0.38    & 5.25    & 10.70                    \\ %\hline
\textbf{Web-Google}                   & 0.87     & 5.10    & 8.08    & 9.90                    \\ %\hline
\textbf{Web-Berk-Stan}                & 0.58     & 7.6     & 23.45  & 123.25                    \\ %\hline
\textbf{Cit-Patent}               & 3.77     & 16.52   & 27.25   & 233.31                  \\ %\hline
\textbf{Com-LiveJournal}          & 3.99     & 34.68   & 75.04   & 513.13                   \\ \hline 
\textbf{ER} & 2.94 & 6.0& 1.59 & 46472.74 \\ %\hline
\textbf{WS} & 8.99 & 36.0 & 20.77 & 786.14 \\ %\hline
\textbf{FF} & 9.0 & 53.1 & 85.11 & 32980.94 \\ \hline
\end{tabular}
   } 
%  \end{minipage}
\end{table}
%\vspace{-1em}

%% file: socialcapital-arXiv.bbl
\begin{thebibliography}{35}
% BibTex style file: aps.bst  (nameyear), 2013-04-23
\ifx \bisbn   \undefined \def \bisbn  #1{ISBN #1}\fi
\ifx \binits  \undefined \def \binits#1{#1} \fi
\ifx \bauthor  \undefined \def \bauthor#1{#1} \fi
\ifx \bjtitle  \undefined \def \bjtitle#1{\textrm{#1}}\fi
\ifx \batitle  \undefined \def \batitle#1{#1} \fi
\ifx \bctitle  \undefined \def \bctitle#1{#1} \fi
\ifx \bvolume  \undefined \def \bvolume#1{\textbf{#1}}\fi
\ifx \byear  \undefined \def \byear#1{#1} \fi
\ifx \bissue  \undefined \def \bissue#1{#1} \fi
\ifx \bfpage  \undefined \def \bfpage#1{#1} \fi
\ifx \blpage  \undefined \def \blpage #1{#1} \fi
\ifx \burl  \undefined \def \burl#1{#1} \fi
\ifx \doiurl  \undefined \def \doiurl#1{#1} \fi
\ifx \betal  \undefined \def \betal{et al.} \fi
\ifx \binstitute  \undefined \def \binstitute#1{#1} \fi
\ifx \beditor  \undefined \def \beditor#1{#1} \fi
\ifx \bpublisher  \undefined \def \bpublisher#1{#1} \fi
\ifx \bbtitle  \undefined \def \bbtitle#1{\textit{#1}} \fi
\ifx \bedition  \undefined \def \bedition#1{#1} \fi
\ifx \bseriesno  \undefined \def \bseriesno#1{#1} \fi
\ifx \blocation  \undefined \def \blocation#1{#1} \fi
\ifx \bsertitle  \undefined \def \bsertitle#1{#1} \fi
\ifx \bsnm \undefined \def \bsnm#1{#1} \fi
\ifx \bsuffix \undefined \def \bsuffix#1{#1} \fi
\ifx \bparticle \undefined \def \bparticle#1{#1} \fi
\ifx \barticle \undefined \def \barticle#1{#1} \fi
\ifx \botherref \undefined \def \botherref #1{#1} \fi
\ifx \url \undefined \def \url#1{#1} \fi
\ifx \bchapter \undefined \def \bchapter#1{#1} \fi
\ifx \bbook \undefined \def \bbook#1{#1} \fi
\ifx \bcomment \undefined \def \bcomment#1{#1} \fi
\ifx \oauthor \undefined \def \oauthor#1{#1} \fi
\ifx \citeauthoryear \undefined \def \citeauthoryear#1{#1} \fi
\ifx \texttildelow  \undefined \def \texttildelow{\symbol{126}} \fi
\def \endbibitem {}
\ifx \bconflocation  \undefined \def \bconflocation#1{#1} \fi

\bibitem[\protect\citeauthoryear{Adler and Kwon}{2002}]{Adler2002Social}
\begin{botherref}
\oauthor{\binits{P.S.} \bsnm{Adler}},
\oauthor{\binits{S.W.} \bsnm{Kwon}},
{Social Capital: Prospects for a New Concept}.
The Academy of Management Review
\textbf{27}(1)
(2002)
\end{botherref}
\endbibitem

\bibitem[\protect\citeauthoryear{Alvarez-Hamelin et~al.}{2006}]{kcoretool06}
\begin{barticle}
\bauthor{\binits{J.I.} \bsnm{Alvarez-Hamelin}},
\bauthor{\binits{A.} \bsnm{Barrat}},
\bauthor{\binits{A.} \bsnm{Vespignani}},
\batitle{{Large Scale Networks Fingerprinting and Visualization Using the
  k-core Decomposition}}.
\bjtitle{Advances in Neural Information Processing Systems}
\bvolume{18},
\bfpage{41}--\blpage{50}
(\byear{2006})
\end{barticle}
\endbibitem

\bibitem[\protect\citeauthoryear{Alvarez{-}Hamelin et~al.}{2008}]{Alvarez08}
\begin{barticle}
\bauthor{\binits{J.I.} \bsnm{Alvarez{-}Hamelin}},
\bauthor{\binits{L.} \bsnm{Dall'Asta}},
\bauthor{\binits{A.} \bsnm{Barrat}},
\bauthor{\binits{A.} \bsnm{Vespignani}},
\batitle{{k-core Decomposition of Internet Graphs: Hierarchies, Self-similarity
  and Measurement Biases}}.
\bjtitle{Networks and Heterogenous Media}
\bvolume{3}(\bissue{2}),
\bfpage{371}
(\byear{2008})
\end{barticle}
\endbibitem

\bibitem[\protect\citeauthoryear{Bakman and Oliver}{2014}]{Leonid14}
\begin{bbook}
\bauthor{\binits{L.} \bsnm{Bakman}},
\bauthor{\binits{A.L.} \bsnm{Oliver}},
\bctitle{{Coevolutionary Perspective of Industry Network Dynamics}}
(\bpublisher{Emerald Group Publishing Ltd.},
\blocation{Bingley,United Kingdom}, \byear{2014}),
pp. \bfpage{3}--\blpage{36}.
\bcomment{Chap. 1}
\end{bbook}
\endbibitem

\bibitem[\protect\citeauthoryear{Bloch et~al.}{2017}]{Bloch16}
\begin{botherref}
\oauthor{\binits{F.} \bsnm{Bloch}},
\oauthor{\binits{M.O.} \bsnm{Jackson}},
\oauthor{\binits{P.} \bsnm{Tebaldi}},
{Centrality Measures in Networks}.
Available at SSRN: https://ssrn.com/abstract=2749124
(2017)
\end{botherref}
\endbibitem

\bibitem[\protect\citeauthoryear{Blondel et~al.}{2008}]{blonde08}
\begin{botherref}
\oauthor{\binits{V.D.} \bsnm{Blondel}},
\oauthor{\binits{J.L.} \bsnm{Guillaume}},
\oauthor{\binits{R.} \bsnm{Lambiotte}},
\oauthor{\binits{E.L.J.S.} \bsnm{Mech}},
{{Fast Unfolding of Communities in Large Networks}}.
Journal of Statistical Mechanics
(2008)
\end{botherref}
\endbibitem

\bibitem[\protect\citeauthoryear{Boldi and Vigna}{2014}]{survey-BoldiV13}
\begin{barticle}
\bauthor{\binits{P.} \bsnm{Boldi}},
\bauthor{\binits{S.} \bsnm{Vigna}},
\batitle{{Axioms for Centrality}}.
\bjtitle{Internet Mathematics}
\bvolume{10}(\bissue{3-4}),
\bfpage{222}--\blpage{262}
(\byear{2014})
\end{barticle}
\endbibitem

\bibitem[\protect\citeauthoryear{Borgatti and Everett}{2006}]{BorgattiE06}
\begin{barticle}
\bauthor{\binits{S.P.} \bsnm{Borgatti}},
\bauthor{\binits{M.G.} \bsnm{Everett}},
\batitle{{A Graph-theoretic Perspective on Centrality}}.
\bjtitle{Social Networks}
\bvolume{28}(\bissue{4}),
\bfpage{466}--\blpage{484}
(\byear{2006})
\end{barticle}
\endbibitem

\bibitem[\protect\citeauthoryear{Burt}{2001}]{burt2001structural}
\begin{bchapter}
\bauthor{\binits{R.S.} \bsnm{Burt}},
\bctitle{{Structural Holes versus Network Closure as Social Capital}},
in \bbtitle{Social Capital: Theory and Research}
(\bpublisher{Aldine de Gruyter},
\blocation{NY, USA}, \byear{2001}),
pp. \bfpage{31}--\blpage{56}
\end{bchapter}
\endbibitem

\bibitem[\protect\citeauthoryear{Cohen}{2008}]{Cohen08}
\begin{botherref}
\oauthor{\binits{J.} \bsnm{Cohen}},
{Trusses: Cohesive Subgraphs for Social Network Analysis}.
NSA:Technical report
(2008)
\end{botherref}
\endbibitem

\bibitem[\protect\citeauthoryear{David and Jon}{2010}]{David2010}
\begin{bbook}
\bauthor{\binits{E.} \bsnm{David}},
\bauthor{\binits{K.} \bsnm{Jon}},
\bbtitle{{Networks, Crowds, and Markets: Reasoning About a Highly Connected
  World}}
(\bpublisher{Cambridge University Press},
\blocation{New York, USA}, \byear{2010})
\end{bbook}
\endbibitem

\bibitem[\protect\citeauthoryear{Freeman}{1979}]{Freeman:SN79}
\begin{barticle}
\bauthor{\binits{L.C.} \bsnm{Freeman}},
\batitle{{Centrality in Social Networks: Conceptual Clarification}}.
\bjtitle{Social Networks}
\bvolume{1}(\bissue{3}),
\bfpage{215}--\blpage{239}
(\byear{1979})
\end{barticle}
\endbibitem

\bibitem[\protect\citeauthoryear{Girvan and Newman}{2002}]{Girvan02}
\begin{barticle}
\bauthor{\binits{M.} \bsnm{Girvan}},
\bauthor{\binits{M.E.J.} \bsnm{Newman}},
\batitle{{Community Structure in Social and Biological Networks}}.
\bjtitle{PNAS}
\bvolume{99}(\bissue{12}),
\bfpage{7821}--\blpage{7826}
(\byear{2002})
\end{barticle}
\endbibitem

\bibitem[\protect\citeauthoryear{Granovetter}{1973}]{granovetter73}
\begin{barticle}
\bauthor{\binits{M.S.} \bsnm{Granovetter}},
\batitle{{The Strength of Weak Ties}}.
\bjtitle{American Journal of Sociology}
\bvolume{78}(\bissue{6}),
\bfpage{1360}--\blpage{1380}
(\byear{1973})
\end{barticle}
\endbibitem

\bibitem[\protect\citeauthoryear{Gupta et~al.}{2016}]{Gupta2016}
\begin{barticle}
\bauthor{\binits{N.} \bsnm{Gupta}},
\bauthor{\binits{A.} \bsnm{Singh}},
\bauthor{\binits{H.} \bsnm{Cherifi}},
\batitle{ˆ{Centrality Measures for Networks with Community Structure}}.
\bjtitle{Physica A: Statistical Mechanics and its Applications}
\bvolume{452},
\bfpage{46}--\blpage{59}
(\byear{2016})
\end{barticle}
\endbibitem

\bibitem[\protect\citeauthoryear{Ilyas and Radha}{2010}]{IlyasR10}
\begin{bchapter}
\bauthor{\binits{M.U.} \bsnm{Ilyas}},
\bauthor{\binits{H.} \bsnm{Radha}},
\bctitle{{A KLT-inspired Node Centrality for Identifying Influential
  Neighborhoods in Graphs}},
in \bbtitle{Proceedings of 44th Annual Conference on Information Sciences and
  Systems},
\byear{2010},
pp. \bfpage{1}--\blpage{7}
\end{bchapter}
\endbibitem

\bibitem[\protect\citeauthoryear{Kang et~al.}{2011}]{KangPST11}
\begin{bchapter}
\bauthor{\binits{U.} \bsnm{Kang}},
\bauthor{\binits{S.} \bsnm{Papadimitriou}},
\bauthor{\binits{J.} \bsnm{Sun}},
\bauthor{\binits{H.} \bsnm{Tong}},
\bctitle{{Centralities in Large Networks: Algorithms and Observations}},
in \bbtitle{Proceedings of SIAM International Conference on Data Mining},
\byear{2011},
pp. \bfpage{119}--\blpage{130}.
\bisbn{978-0-898719-92-5}
\end{bchapter}
\endbibitem

\bibitem[\protect\citeauthoryear{Koschade}{2006}]{Koschade2006}
\begin{barticle}
\bauthor{\binits{S.A.} \bsnm{Koschade}},
\batitle{{A Social Network Analysis of Jemaah Islamiyah: The Applications to
  Counter-Terrorism and Intelligence}}.
\bjtitle{Studies in Conflict and Terrorism}
\bvolume{29}(\bissue{6}),
\bfpage{559}--\blpage{575}
(\byear{2006})
\end{barticle}
\endbibitem

\bibitem[\protect\citeauthoryear{Landherr et~al.}{2010}]{Landherr2010}
\begin{barticle}
\bauthor{\binits{A.} \bsnm{Landherr}},
\bauthor{\binits{B.} \bsnm{Friedl}},
\bauthor{\binits{J.} \bsnm{Heidemann}},
\batitle{{A Critical Review of Centrality Measures in Social Networks}}.
\bjtitle{{Business and Infornmation Systems Engineering}}
\bvolume{2}(\bissue{6}),
\bfpage{371}--\blpage{385}
(\byear{2010})
\end{barticle}
\endbibitem

\bibitem[\protect\citeauthoryear{Li et~al.}{2015}]{li2015}
\begin{botherref}
\oauthor{\binits{C.} \bsnm{Li}},
\oauthor{\binits{Q.} \bsnm{Li}},
\oauthor{\binits{P.V.} \bsnm{Mieghem}},
\oauthor{\binits{H.E.} \bsnm{Stanley}},
\oauthor{\binits{H.} \bsnm{Wang}},
{Correlation Between Centrality Metrics and Their Application to the Opinion
  Model}.
The European Physical Journal B
\textbf{88}(65)
(2015)
\end{botherref}
\endbibitem

\bibitem[\protect\citeauthoryear{Lin}{2008}]{Lin2008}
\begin{bchapter}
\bauthor{\binits{N.} \bsnm{Lin}},
\bctitle{{A Network Theory of Social Capital}},
in \bbtitle{The Handbook of Social Capital}
(\bpublisher{Oxford University Press},
\blocation{Oxford , New York, USA}, \byear{2008}),
pp. \bfpage{50}--\blpage{69}.
\bcomment{Chap. 2}
\end{bchapter}
\endbibitem

\bibitem[\protect\citeauthoryear{Nahapiet and Ghoshal}{1998}]{ghoshalsc-98}
\begin{barticle}
\bauthor{\binits{J.} \bsnm{Nahapiet}},
\bauthor{\binits{S.} \bsnm{Ghoshal}},
\batitle{{Social Capital, Intellectual Capital, and the Organizational
  Advantage}}.
\bjtitle{The Academy of Management Rev}
\bvolume{23}(\bissue{2}),
\bfpage{242}--\blpage{266}
(\byear{1998})
\end{barticle}
\endbibitem

\bibitem[\protect\citeauthoryear{Nick et~al.}{2013}]{Nick2013}
\begin{bchapter}
\bauthor{\binits{B.} \bsnm{Nick}},
\bauthor{\binits{C.} \bsnm{Lee}},
\bauthor{\binits{P.} \bsnm{Cunningham}},
\bauthor{\binits{U.} \bsnm{Brandes}},
\bctitle{{Simmelian Backbones: Amplifying Hidden Homophily in Facebook
  Networks}},
in \bbtitle{Proceedings of IEEE/ACM International Conference on Advances in
  Social Networks Analysis and Mining},
\byear{2013},
pp. \bfpage{525}--\blpage{532}
\end{bchapter}
\endbibitem

\bibitem[\protect\citeauthoryear{Ortmann and Brandes}{2014}]{Ortmann2014}
\begin{bchapter}
\bauthor{\binits{M.} \bsnm{Ortmann}},
\bauthor{\binits{U.} \bsnm{Brandes}},
\bctitle{{Triangle Listing Algorithms: Back from the Diversion}},
in \bbtitle{Proceedings of the Sixteenth Workshop on Algorithm Engineering and
  Experiments (ALENEX)},
\byear{2014},
pp. \bfpage{1}--\blpage{8}
\end{bchapter}
\endbibitem

\bibitem[\protect\citeauthoryear{Phillip}{1987}]{bonacich1987power}
\begin{barticle}
\bauthor{\binits{B.} \bsnm{Phillip}},
\batitle{{Power and Centrality: A Family of Measures}}.
\bjtitle{American Journal of Sociology}
\bvolume{92}(\bissue{5}),
\bfpage{1170}--\blpage{1182}
(\byear{1987})
\end{barticle}
\endbibitem

\bibitem[\protect\citeauthoryear{Putnam}{2002}]{Putnam2002}
\begin{bbook}
\bauthor{\binits{R.D.} \bsnm{Putnam}},
\bbtitle{Democracies in Flux: The Evolution of Social Capital in Contemporary
  Society}
(\bpublisher{Oxford University Press},
\blocation{New York}, \byear{2002})
\end{bbook}
\endbibitem

\bibitem[\protect\citeauthoryear{Qi et~al.}{2013}]{Laplacian13}
\begin{barticle}
\bauthor{\binits{X.} \bsnm{Qi}},
\bauthor{\binits{R.D.} \bsnm{Duval}},
\bauthor{\binits{K.} \bsnm{Christensen}},
\bauthor{\binits{E.} \bsnm{Fuller}},
\bauthor{\binits{A.} \bsnm{Spahiu}},
\bauthor{\binits{Q.} \bsnm{Wu}},
\bauthor{\binits{Y.} \bsnm{Wu}},
\bauthor{\binits{W.} \bsnm{Tang}},
\bauthor{\binits{C.} \bsnm{Zhang}},
\batitle{{Terrorist Networks, Network Energy and Node Removal: A New Measure of
  Centrality Based on Laplacian Energy}}.
\bjtitle{Social Networking}
\bvolume{2},
\bfpage{19}--\blpage{31}
(\byear{2013})
\end{barticle}
\endbibitem

\bibitem[\protect\citeauthoryear{Qi et~al.}{2015}]{Qi2015}
\begin{barticle}
\bauthor{\binits{X.} \bsnm{Qi}},
\bauthor{\binits{E.} \bsnm{Fuller}},
\bauthor{\binits{R.} \bsnm{Luo}},
\bauthor{\binits{C.} \bsnm{Zhang}},
\batitle{{A Novel Centrality Method for Weighted Networks based on the
  Kirchhoff Polynomial}}.
\bjtitle{Pattern Recognition Letters}
\bvolume{58},
\bfpage{51}--\blpage{60}
(\byear{2015})
\end{barticle}
\endbibitem

\bibitem[\protect\citeauthoryear{Radicchi et~al.}{2004}]{Radicchi2004}
\begin{barticle}
\bauthor{\binits{F.} \bsnm{Radicchi}},
\bauthor{\binits{C.} \bsnm{Castellano}},
\bauthor{\binits{F.} \bsnm{Cecconi}},
\bauthor{\binits{V.} \bsnm{Loreto}},
\bauthor{\binits{D.} \bsnm{Parisi}},
\batitle{{Defining and Identifying Communities in Networks}}.
\bjtitle{Proceedings of the National Academy of Sciences}
\bvolume{101}(\bissue{9}),
\bfpage{2658}--\blpage{2663}
(\byear{2004})
\end{barticle}
\endbibitem

\bibitem[\protect\citeauthoryear{Shi et~al.}{2007}]{Adamic07}
\begin{barticle}
\bauthor{\binits{X.} \bsnm{Shi}},
\bauthor{\binits{L.A.} \bsnm{Adamic}},
\bauthor{\binits{M.J.} \bsnm{Strauss}},
\batitle{{Networks of Strong Ties}}.
\bjtitle{Physica A: Statistical Mechanics and its Applications}
\bvolume{378}(\bissue{1}),
\bfpage{33}--\blpage{47}
(\byear{2007})
\end{barticle}
\endbibitem

\bibitem[\protect\citeauthoryear{Subbian et~al.}{2013}]{SoCap2014}
\begin{bchapter}
\bauthor{\binits{K.} \bsnm{Subbian}},
\bauthor{\binits{D.} \bsnm{Sharma}},
\bauthor{\binits{Z.} \bsnm{Wen}},
\bauthor{\binits{J.} \bsnm{Srivastava}},
\bctitle{{Finding Influencers in Networks using Social Capital}},
in \bbtitle{Proceedings of IEEE/ACM International Conference on Advances in
  Social Networks Analysis and Mining},
\byear{2013}
\end{bchapter}
\endbibitem

\bibitem[\protect\citeauthoryear{Valente et~al.}{2008}]{correlated-centrality}
\begin{barticle}
\bauthor{\binits{T.W.} \bsnm{Valente}},
\bauthor{\binits{K.} \bsnm{Coronges}},
\bauthor{\binits{C.} \bsnm{Lakon}},
\bauthor{\binits{E.} \bsnm{Costenbader}},
\batitle{{How Correlated are Network Centrality Measures?}}
\bjtitle{Connections}
\bvolume{28}(\bissue{1}),
\bfpage{16}--\blpage{26}
(\byear{2008})
\end{barticle}
\endbibitem

\bibitem[\protect\citeauthoryear{Wang and Cheng}{2012}]{wang-vldb12}
\begin{barticle}
\bauthor{\binits{J.} \bsnm{Wang}},
\bauthor{\binits{J.} \bsnm{Cheng}},
\batitle{{Truss Decomposition in Massive Networks}}.
\bjtitle{{Proceedings of VLDB Endowment}}
\bvolume{5}(\bissue{9}),
\bfpage{812}--\blpage{823}
(\byear{2012})
\end{barticle}
\endbibitem

\bibitem[\protect\citeauthoryear{Wang et~al.}{2015}]{Wang15}
\begin{barticle}
\bauthor{\binits{M.} \bsnm{Wang}},
\bauthor{\binits{C.} \bsnm{Wang}},
\bauthor{\binits{J.X.} \bsnm{Yu}},
\bauthor{\binits{J.} \bsnm{Zhang}},
\batitle{{Community Detection in Social Networks: An In-depth Benchmarking
  Study with a Procedure-oriented Framework}}.
\bjtitle{Proceedings of VLDB Endowment}
\bvolume{8}(\bissue{10}),
\bfpage{998}--\blpage{1009}
(\byear{2015})
\end{barticle}
\endbibitem

\bibitem[\protect\citeauthoryear{Xie and Szymanski}{2013}]{LabelRank13}
\begin{bchapter}
\bauthor{\binits{J.} \bsnm{Xie}},
\bauthor{\binits{B.K.} \bsnm{Szymanski}},
\bctitle{{LabelRank: A Stabilized Label Propagation Algorithm for Community
  Detection in Networks}},
in \bbtitle{Proceedings of the 2nd {IEEE} Network Science Workshop},
\byear{2013},
pp. \bfpage{138}--\blpage{143}
\end{bchapter}
\endbibitem

\end{thebibliography}
